\newif\iffinal
\finalfalse
\iffinal
\documentclass[compsoc]{IEEEtran} %
\else
\documentclass[11pt,onecolumn,draftcls]{IEEEtran} %
\fi
\usepackage{multirow}
\usepackage{amsfonts,amssymb,amsmath,amsthm,mathrsfs}

\usepackage[ruled,vlined,linesnumbered]{algorithm2e}
\SetKwProg{Fn}{function}{}{end}
\SetKwProg{struct}{struct}{ contains}{end}
\SetKw{variable}{variable}
\SetCommentSty{small}
\SetKw{return}{return}
\SetKw{throw}{throw}
\SetKw{event}{event}
\SetKw{Initvar}{Initialized variables}

\usepackage{xcolor}
\usepackage{ifpdf}
\ifpdf
\usepackage[]{graphicx}
\else
\usepackage[draft]{graphicx}
\fi
\usepackage{array}
\usepackage{float}
\usepackage[caption = false]{subfig}
\usepackage{caption}
\usepackage{url,cite}
\usepackage{breakurl}
\usepackage{changepage}
\usepackage{listings}
\usepackage{makecell}

\newtheorem{theorem}{Theorem}

\def\bbA{\mathbb{A}}

\def\calF{\mathcal{F}}

\def\calL{\mathcal{L}}

\def\calO{\mathcal{O}}
\def\calP{\mathcal{P}}
\def\calR{\mathcal{R}}
\def\calT{\mathcal{T}}

\def\hash{\mathsf{Hash}}

\def\id{\mathsf{ID}}

\def\B{\mathsf{B}}

\def\symenc{\mathsf{Sym}\mbox{-}\mathsf{Enc}}
\def\symdec{\mathsf{Sym}\mbox{-}\mathsf{Dec}}
\def\asymenc{\mathsf{Asym}\mbox{-}\mathsf{Enc}}
\def\asymdec{\mathsf{Asym}\mbox{-}\mathsf{Dec}}
\def\cpabeenc{\mathsf{CPABE}\mbox{-}\mathsf{Enc}}
\def\cpabedec{\mathsf{CPABE}\mbox{-}\mathsf{Dec}}

\def\sign{\mathsf{Sign}}
\def\ecrecover{\mathsf{ecrecover}}

\def\para#1{\vskip .3ex\noindent\textit{#1}}
\def\maxtran{\mathsf{MAX}\_\mathsf{TRAN}}

\def\staffMember{\mathsf{staffMember}}
\def\addStaffMember{\mathsf{addStaffMember}}
\def\getAttributes{\mathsf{getAttributes}}
\def\verifyRequest{\mathsf{verifyRequest}}
\def\ecrecover{\mathsf{ecrecover}}
\def\LogAnnounce{\mathsf{LogAnnounce}}
\def\satisfy{\mathsf{satisfy}}
\def\SMR{\mathsf{SMR}}
\def\AVPA{\mathsf{AVPA}}
\def\GK{\mathsf{GK}}
\def\addkey{\mathsf{addKey}}
\def\LogKeys{\mathsf{LogKeys}}
\def\staffID{\mathsf{staffID}}
\def\staffAttributes{\mathsf{staffAttributes}}
\def\upBound{\mathsf{upBound}}
\def\Map{\mathsf{Map}}
\def\address{\mathsf{\_address}}
\def\id{\mathsf{\_id}}
\def\attributes{\mathsf{\_attributes}}
\def\setOfCertifiers{\mathsf{setOfCertifiers}}
\def\msgsender{\mathsf{msg.sender}}
\def\staffAddress{\mathsf{\_staffAddress}}
\def\msg{\mathsf{\_msg}}
\def\signer{\mathsf{signer}}
\def\pub{\mathsf{pub}}
\def\pr{\mathsf{pr}}
\def\fetched{\mathsf{fetched}}
\def\ipfsAddress{\mathsf{\_ipfsAddress}}
\def\setOfIssuers{\mathsf{setOfIssuers}}
\def\PK{\mathsf{PK}}
\def\MK{\mathsf{MK}}
\def\CT{\mathsf{CT}}
\def\SK{\mathsf{SK}}
\def\F{\mathsf{F}}
\def\sk{\mathsf{sk}}
\def\R{\mathsf{R}}
\def\A{\mathsf{A}}
\def\ER{\mathsf{ER}}
\def\Esk{\mathsf{Esk}}
\def\ESK{\mathsf{ESK}}

\usepackage[normalem]{ulem}
\definecolor{revision}{rgb}{1,0,0}

\def\add#1{\textcolor{blue}{#1}}

\def\para#1{\smallskip\noindent\textbf{#1}:~~}
\begin{document}
\title{$d$-MABE: Distributed Multilevel Attribute-Based EMR Management and Applications}
\author{Ehab Zaghloul, Tongtong Li, Matt Mutka and Jian Ren%
\thanks{The authors are with Michigan State University, East Lansing, MI 48824-1226, Email: \{ebz, tongli, mutka, renjian\}@msu.edu}}

\maketitle
\begin{abstract}
Current systems used by medical institutions for the management and transfer of Electronic Medical Records (EMR) can be vulnerable to security and privacy threats.
In addition, these centralized systems often lack interoperability and give patients limited or no access to their own EMRs. 
In this paper, we propose a novel distributed data sharing scheme that applies the security benefits of blockchain to handle these concerns. 
With blockchain, we incorporate smart contracts and a distributed storage system to alleviate the dependence on the record-generating institutions to manage and share patient records. 
To preserve privacy of patient records, we implement our smart contracts as a method to allow patients to verify attributes prior to granting access rights. 
Our proposed scheme also facilitates selective sharing of medical records among staff members that belong to different levels of a hierarchical institution.
We provide extensive security, privacy, and evaluation analyses to show that our proposed scheme is both efficient and practical.
\end{abstract}

\section{Introduction}
Medical record-sharing between medical institutions can help improve medical diagnoses, treatment decisions, and enhance the overall patient experience. 
However, the current methods used to manage and share these records can be inefficient or limited due to interoperable systems. 
In addition, legal requirements associated with the sharing of patient records makes patient data especially vulnerable to security and privacy threats.

The history of generating medical records in the U.S. dates back to the 1800s. 
These records included sections such as medical history, family history, patient information, and admitted samples. 
Medical records during this time would be scattered between separate volumes maintained by different branches of service such as surgery and outpatient services as there was no way to link them. 
In 1907, a new system was introduced at St. Mary’s Hospital and the Mayo Clinic to address the issue of disorganized patient records. 
In this system, each new patient is assigned a clinic number, and all data associated would be combined into one patient record. 
By the early to mid-1900s, paper medical records became more complex, eventually becoming difficult to maintain. 
These paper records also presented difficulties such as availability to transfer between hospitals and illegibility from handwriting. 

Over the last fifty years, there has been an effort to transition from paper to EMRs, primarily due to extensive requirements of government programs such as Medicare and Medicaid. 
In 2009, the U.S. invested \$19 billion under the American Recovery and Reinvestment Act for digitization of health records~\cite{fernandez2013security}.  
By 2011, over half of the physicians reported using an electronic system for health record management~\cite{gillum2013papyrus}. 
 
In the U.S., regulations on patient records are governed by the Health Insurance Portability and Accountability Act of 1996 (HIPAA)~\cite{hhs}. 
Under HIPAA, the standards for privacy of Individually Identifiable Health Information (IIHI) sets the national standard for health information privacy protection. 
One of the key protections under HIPAA is to preserve the privacy of the patient records by limiting their data released to the minimal amount for any intended purpose. 
These protections allow the legal movement of sensitive health data needed for the successful administration of services.  Under HIPAA, data encryption is recognized as an \textit{addressable} requirement and a HIPAA governed entity can determine whether it will encrypt
data or implement an appropriate alternative. 
From a security outlook, inferior alternatives to data encryption can severely impede on patient privacy.

While HIPAA sets a guideline for secure and private use of patient health data, it does not define how these systems are to be developed and applied. 
As a result of this, many centralized and private applications have come to market as ways to manage EMRs.
Because these systems are often built on proprietary technology, lack of interoperability between medical institutions is a major concern.
Reuters reported in 2015 that
only 30\% of U.S. hospitals are able to locate, send, and/or receive EMRs for patients that were treated somewhere else~\cite{reuters}. 
This means that, in majority, the burden of transferring health records lies on patients, often times in printed copies from one medical institution to another.
In other cases, patient records are transferred via fax or systems similar to electronic mail. 
These common transfer methods are inefficient and could jeopardize patient privacy and data security. 
There is a lack of an auditable trail of data transfer, and there is no way to preclude who will view a patient record on the other end of a fax. 
Lastly in cases where there is no likely method of data transfer or data cannot be trusted, physicians may resort to duplicate testing causing resource waste.  

In addition to the lack of interoperability, the current electronic record systems are also susceptible to threats such as Distributed Denial of Service (DDoS) attacks~\cite{lau2000distributed} and social engineering attacks~\cite{anderson2008security}. 
Beckers Hospital reported that 18\% of all claimed cyber-attacks in 2017 were attributed to the healthcare sector.
This comprised 28\% of the total cost for all data breaches, showing the staggering cost of a healthcare data breach~\cite{becker}.

As a result of growing security concerns associated with the trust of a third party, in 2009, the blockchain technology was introduced with the evolution of the digital cryptocurrency, Bitcoin~\cite{Nakamoto08Bitcoin}.
One of the main goals of this project was to eliminate the need of a trusted third party, such as banks, to facilitate payments between individuals.
Bitcoin uses the blockchain technology, a cryptographically secure public ledger of transactions between all involved parties that is maintained over a peer-to-peer network.
Since the inception of Bitcoin and the success of its technology, various blockchain-based systems began to evolve that strive to reinforce similar features such as security, privacy, and efficiency.

Based on a concept initially introduced by Nick Szabo in 1994~\cite{szabo1996smart}, smart contracts were later incorporated into blockchain-based systems.
Smart contracts are self-executing pieces of code that can digitally enforce and facilitate verified negotiations of contracts between two participating individuals without the need of a trusted third party.
Today, they continue to appear in systems, for example, Ethereum~\cite{wood2014ethereum}, that aim to provide services beyond simple payments.
Using such smart contracts, users can design and customize their own systems on top of blockchains such as Ethereum.
Similar to payment transactions, the deployment and execution of smart contracts is immutable, irreversible and can be tracked over the public blockchain.
Therefore, smart contracts are well suited to facilitate data management and sharing.
They can be used by data owners to define access control mechanisms that grant certain data users access to the data.

In this paper, we propose a novel distributed data sharing scheme that leverages blockchain and smart contracts.
Our proposed scheme aims to provide security, privacy, and efficiency to the data owners and users of the system.
In terms of security, patient records are stored and protected from malicious data users.
Our proposed scheme can also facilitate private and selective record-sharing; allowing patients to share selectively different parts of their record with distinct data users based on their privacy preferences.
Using our proposed scheme, the record-generating institutions can be entirely eliminated from the record-sharing process.

The work presented in this paper expands significantly on the model presented in \cite{zaghloul18}.
In contrast, our proposed scheme is distributed in terms of data management and storage, and empowers the patients over their records whereby minimizing dependency on the record-generating institutions.
The main contributions of this paper can be summarized as follows:

(i) We develop a novel distributed electronic medical record-sharing scheme for patients that is secure, preserves the privacy of patient data, and is efficient.

(ii) We propose a distributed method for verifying medical institution staff member attributes over the blockchain before issuing them access keys.
This method can help minimize trust and dependencies on key-issuers when verifying attributes of staff members.

(iii) We conduct comprehensive security and privacy analyses of the proposed scheme. 

(iv) We implement our proposed scheme using smart contracts and deploy them over the Ethereum blockchain for performance evaluation and numerical demonstration.

The rest of this paper is organized as follows.
In Section~\ref{sec:relatedwork}, the related work is reviewed. 
In Section~\ref{sec:prelim}, preliminaries are introduced that summarize key concepts used in this research. 
Next, in Section~\ref{sec:problem}, the problem formulation is described outlining the system model and our design goals.
In Section~\ref{sec:proposed}, our proposed scheme is presented in detail outlining the proposed algorithms. 
Following that, in Section~\ref{sec:security} we formally prove the security of our proposed scheme.
In Section~\ref{sec:evaluation}, we present a performance and application analysis that is supported with numerical results. 
Finally, in Section~\ref{sec:conclusion}, a conclusion is drawn to summarize our research.

\section{Related Work}\label{sec:relatedwork}

In 2015, a decentralized data management scheme was introduced that facilitated access-control management over a blockchain~\cite{zyskind2015decentralizing}.
In this system, the actual data records are stored in an off-blockchain storage while pointers to these records are maintained by a key-value storage over the blockchain.
This solution helps simplify the amount of data processed on the blockchain.
However, the method used to define access policies in this scheme does not consider hierarchical data sharing. A user that desires to share files selectively among multiple users in a hierarchy will have to define several access policies that could become a complex problem as the number of users increases drastically.

In 2016, MedRec~\cite{azaria2016medrec}, the first functional electronic medical record-sharing system built on some concepts from~\cite{zyskind2015decentralizing} was introduced. 
This work builds on three Ethereum~\cite{wood2014ethereum} smart contracts that manage the authentication, confidentiality, and accountability during the data sharing process.
In this system, the primary entities involved in maintaining the blockchain are the parties interested in gaining data, such as researchers and public health authorities.
In return, the institutions are rewarded with access to aggregate and anonymized data.
However, the success of such a system is dependent on the participation of entities that maintain the system in return for data.
In addition, similar to~\cite{zyskind2015decentralizing}, MedRec does not consider hierarchical data sharing.

In 2017, another functioning electronic medical record-sharing scheme was presented to provide a secure solution using the blockchain~\cite{dubovitskaya2017secure}.
The system uses a cloud-based storage system to store the medical records.
With centralized storage, the system becomes liable to a single point of failure.
Similar to MedRec, this work builds over a private blockchain, which is a monitored blockchain where each node involved in maintaining consensus is known.
Studies such as~\cite{dinh2018untangling} and~\cite{dinh2017blockbench} present possible methods to evaluate private blockchains and potential concerns and vulnerabilities.
They show that there is a trade-off between performance and security.

In 2018, a study was conducted to discuss several open problems in order to use blockchains for data management and analytics~\cite{vo2018research}.
The study shows that more research is required in order to leverage the capabilities of current data management systems into blockchain-based models.
In addition, the study reflects the importance of improving the security and privacy countermeasures.

\section{Preliminaries}\label{sec:prelim}

\subsection{Ciphertext Policy Attribute-Based Encryption}\label{subsec:cp-abe}
Ciphertext Policy Attribute-Based Encryption (CP-ABE) is a fine-grained access control and encryption scheme~\cite{bethencourt2007ciphertext}.
It allows users to share their data selectively by using access policies that are integrated into the ciphertexts during encryption.
CP-ABE formally divides the process into four main functions.

(i) $\mathsf{Setup}(1^\kappa)$:
a probabilistic function with input security parameter $\kappa$. 
The function is performed by a key-issuer to generate a public key $\PK$ and master key $\MK$.

(ii) $\mathsf{Key Generation}(\MK,\bbA)$:
a probabilistic function carried out by a key-issuer to generate a unique secret key $\SK$ for a data user based on the $\MK$ and a set of attributes $\bbA = \{\A_1,\A_2\dots,\A_n\}$ possessed by the user.

(iii) $\mathsf{Encryption}(\PK$,$\mathsf{m}$,$\calT)$:
a probabilistic function carried out by the data owner to encrypt message $\mathsf{m}$ under an access policy $\calT$ and generate ciphertext $\CT$.

(iv) $\mathsf{Decryption}(\CT,\SK)$:
a deterministic function carried out by the data user to decrypt a ciphertext $\CT$ using the uniquely generated secret key $\SK$ for the user.

\subsection{Privilege-based Multilevel Organizational Data-sharing}\label{subsec:pmod}
Privilege-based Multilevel Organizational Data-sharing scheme (P-MOD)~\cite{zaghloul2018p} is an extension of the CP-ABE that handles data sharing in complex hierarchical organizations.
P-MOD partitions a data file into multiple segments based on user privileges and data sensitivity. 
Each segment of the data record is shared according to user privileges and the set of attributes that satisfy the hierarchical access policies.

\para{Privilege-based Access Structure}
The privilege-based access structure divides the data users of an organization into a hierarchy that consists of $k$ levels, $\{\calL_1,\calL_2,\dots,\calL_k\}$.
Each level is associated with an access policy, $\{ \calT_1,\calT_2,\dots,\calT_k \}$.
An access policy $\calT_i$ specifies a set of rules defined using logic gates and the different sets of attributes that can satisfy these rules.
Data users in possession of a correct set of attributes that can satisfy a certain access policy $\calT_i$ belong to level $\calL_i$.

\para{Data Partitioning and Encryption}
The data owner partitions a data file $\calF$ into a set of $k$ record segments, that is $\calF=\{\F_1,\F_2,\dots,\F_k \}$. 
Segment $\F_1$ contains the most sensitive information of $\calF$ that is to be shared with data users belonging to level $\calL_1$.
Segment $\F_k$ contains the least sensitive information of $\calF$ that is to be shared with all data users belonging to any level.
Next, the data owner generates a set of secret keys $\{\sk_1,\sk_2,\dots,\sk_k\}$ to encrypt the corresponding segments of $\calF$.
The key $\sk_1$ is randomly selected by the data owner.
The remaining keys are then derived using a cryptographic hash function $\hash$ as follows
\begin{equation} \label{keys}
\sk_{i+1} = \hash(\sk_i).
\end{equation}

A privileged data user that satisfies access policy $\calT_i$ belongs to level $\calL_i$ and is granted key $\sk_i$.
Given $\sk_i$, the data user can derive keys $\{\sk_{i+1},\dots,\sk_k\}$ belonging to levels $\{\calL_{i+1},\dots,\calL_k\}$ as defined in equation~\eqref{keys}.
However, given the properties of $\hash$ function, $\sk_i$ cannot be used to derive any of the keys $\{\sk_1,\dots,\sk_{i-1}\}$.
Each $\F_i \in \calF$ is then encrypted with its corresponding generated secret key $\sk_i$.
Finally, each $\sk_i$ is encrypted with CP-ABE under its corresponding access policy $\calT_i$.

\subsection{Blockchain}
The blockchain is a public ledger that stores all the cryptographically processed transactions performed over a peer-to-peer network.
Initially implemented by Bitcoin~\cite{Nakamoto08Bitcoin}, it provides its users with transaction confirmations to track ownership rights of the Bitcoin (BTC) cryptocurrency.
The Bitcoin blockchain is based on a distributed consensus, \emph{Proof-of-Work} (PoW), that allows any past and present online transaction to be verified.
It consists of blocks $\{ \B_0,\B_1,\cdots,\B_n \}$, each carrying a set of validated transactions, where $\B_0$ represents the first block and $\B_n$ represents the most recent block attached to the blockchain.
Each block $\B_i$ incorporates the cryptographic hash of its preceding block $\B_{i-1}$ to form the complete blockchain.

In Bitcoin, transactions are generally limited to transferring ownership rights from one user to the other.
Recently, Ethereum was implemented that extended the transaction capabilities of Bitcoin.
Ethereum is an open source blockchain-based network that provides its users with a platform to build, deploy and run decentralized applications~\cite{buterin2017ethereum,wood2014ethereum}.
The network nodes offer a decentralized Turning-complete virtual machine. 
The Ethereum Virtual Machine (EVM) can execute scripts uploaded by the users (referred to as \textit{smart contracts}) over the public and non-trusted network nodes.
In addition, transactions executed by the EVM are attached with receipts that may contain log entries.
These logs represent the results of \textit{events} having fired from the smart contracts.

\para{Smart Contracts}
Smart contracts are pieces of code that enable users to write their own rules for ownership, transaction formats, and state transition functions.
This means that two parties can digitally interact following a set of customized rules without the need of a trusted third party to secure the transaction.
The deployment and/or interaction with smart contracts are immutable, permanent, and irreversible.
The process of deploying or interacting with an already deployed smart contract over the blockchain requires the users to connect to the peer-to-peer network through a client.

\para{Network client}
A client is any node within the peer-to-peer network that can interact with the blockchain.
Interactions may include parsing transactions, verifying transactions/blocks, deploying or executing smart contracts and everything related.
In general, each client generates its own unique public and private key pair~\cite{baumgart2007s} before establishing connections with other nodes and interacting with the blockchain.
The public key $\pub$ is used to generate a public $\mathsf{address}$ for the node within the network, where $\mathsf{address} = \hash(pub)$.
The private key $\pr$ is maintained securely by the node and is used to sign transactions and prove ownership rights.

\para{The EVM}
In Ethereum, the network nodes manage the EVM and are responsible for deploying/executing smart contracts in return for a reward that keeps them incentivized.
Rewards are paid in ETH (the Ethereum currency) by the nodes deploying or transacting with an already deployed smart contract.
To calculate the rewards, the EVM uses a special unit known as \textit{gas} that represents the amount of computational work performed by the network nodes.
More complex transactions require more gas to be executed.
Gas also helps prevent network spamming attempts that try to exhaust the computational power of the network nodes by making these attacks expensive.
The transacting nodes then specify a \textit{gas price}, the amount of ETH they are willing to pay for each unit of gas.
In cases where smart contracts contain buggy code or error, execution may continuously go on while consuming more funds.
Therefore, to prevent such situations, the nodes also define a \textit{gas limit}, a maximum amount of gas they are willing to pay when transacting.
If the gas limit runs out before execution is complete, the execution is halted and the transaction fails.
This means, if the gas limit and/or gas price defined by a transacting node is too low, the transaction may not be executed.

\para{Types of blockchains}
In the majority of blockchain platforms, users can choose to set up three different blockchain environments.
The public blockchain is the decentralized version in that any participating client can interact and participate in the consensus process.
Consortium blockchains are partially decentralized where the consensus process is performed by a pre-selected set of client nodes.
Finally, private blockchains are entirely centralized in which one organization has the complete authority to manipulate the blockchain as it desires.

\subsection{InterPlanetary File System} \label{ipfs}
InterPlanetary File System (IPFS) \cite{ipfs} is a peer-to-peer (P2P) network used to store and share content-addressable data over the network nodes in a distributed manner.
The network runs a stack of protocols that are responsible for storing and accessing the data stored over it.

\para{Identities} 
To join the P2P network, nodes first generate unique node identities before establishing connections.
IPFS node identity generation is based on S/Kadmelia~\cite{baumgart2007s} where each node solves a hard Proof-of-Work (PoW)~\cite{Nakamoto08Bitcoin} crypto-puzzle before generating a public and private key pair~\cite{baumgart2007s}.
The public key is then cryptographically hashed to generate the identity of the node, $\mathsf{NodeId}$. 
Nodes that are willing to accept incoming connection requests verify that the $\mathsf{NodeId}$ matches the hash computation of the public key of the requesting node.

\para{Network}
The network layer in IPFS manages the established connections between the connected peers.
It is designed to run on top of any network, hence, it includes various underlying protocols. 
The network layer eliminates the need of IP addresses allowing IPFS to be used in overlay networks.
It relies on a customized address format that is stored in $\mathsf{multiaddr}$ formatted byte strings.
The purpose of these strings is to express the details of the customized addresses used by the underlying network.

\para{Routing}
The routing layer keeps track of the addresses of nodes within the network and the data they store.
It uses a Distributed Hash Table (DHT) inspired by S/Kademlia \cite{baumgart2007s} and Coral \cite{freedman2004democratizing} to identify the data stored over any node of the network.
Data of size 1KB or less is stored directly on the DHT.
For data of larger sizes, the DHT stores the $\mathsf{NodeIds}$ of the nodes storing the data.

\para{Exchange}
IPFS uses BitSwap, a data exchange protocol inspired by BitTorrent \cite{bittorrent}, to exchange data between nodes.
Bitswap requires each node to share a $\mathsf{want\_list}$ (list of data it is searching for) and a $\mathsf{have\_list}$ (list of data it stores) with its peer nodes.
Nodes are incentivized to share data with their peers even when they are not requesting data through a credit system called BitSwap.
Bitswap keeps a record of how data has been shared previously between nodes to generate a credit balance for each node.
The credit balance is an indicator of how much data a node has shared and can be used by a node to decide whether it should respond to requests or ignore them.

\para{Objects}
By default, IPFS partitions data files into segments of 256KB and then cryptographically hashes each segment.
The resulting digests are used to reference the location of the corresponding segments over the IPFS.
IPFS also forms links between corresponding digests to allow complete data recovery.
The links of a single segment are stored in an IPFS link array as cryptographic hashes of the target data segments.
That means that if a node gains access to a particular segment of a data file, it can also recursively gain access to all of its linked segments.  
However, it cannot obtain access to the unlinked segments of the same file.
IPFS uses a Merkle Directed Acyclic Graph (Merkle DAG) \cite{dag} to maintain these reference links.
This method provides data tampering resistance and single storage of duplicated segments.

\para{Files}
Expanding on GIT \cite{git}, IPFS can store versioned file systems on top of the Merkle DAG \cite{dag}.
An object in this model consists of a variable sized data block (also referred to as a blob), a list containing a collection of blocks, a tree that combines the blocks and lists, and a commit that reflects a snapshot in the version history tree.

\para{Naming}
IPFS uses an InterPlanetary Naming System (IPNS) that builds on Self-certifying File System~\cite{mazieres2000self}.
IPNS is a mutable naming system to immutable data content maintained by the Merkle DAG.
It allows data owners to upload modified versions of their data and redirect the reference links from the previous versions to the updated versions.

\section{Problem Formulation}\label{sec:problem}

Upon visiting a medical institution, a patient may interact with various staff members during the treatment process.
If those staff members can efficiently access the previous medical record(s) of the patient generated by other institutions, they can provide the patient with an enhanced diagnosis and treatment decisions.
For example, an ER physician may learn severe drug allergies or current medication for an unresponsive patient being rushed into the ER from his/her previous records. 
As a result, the physician can prevent iatrogenic illnesses caused by administering inappropriate medication to that patient or harmful drug interactions. 
By accessing critical medical information, patient safety is enhanced while saving time and resources.
However, for the sake of patient privacy, only necessary staff members should be given permission to access patient record(s) from the medical history determined by the profile of the patient. 

In this paper, we denote a record generated by an institution for a patient as $\calR$.
We denote the record attributes as $\{ \R_j \in \calR ~|~ 1 \leq j \leq n \}$, where $n$ is the maximum number of attributes within the record.
The corresponding attribute values of the record generated for the patient can be denoted as $\{a(\R_j) ~|~ 1 \leq j \leq n \}$.

In the current systems, records are often variable in record attributes and may be blocked or intentionally delayed by competitive record-generating institutions. 
To prevent this, record-generating institutions should have minimal control over the records of patients whereby empowering patients over their own records. 
Given the option to empower patients over their own health data, patients may have different attitudes in terms of sensitivity of their record attributes. 
The challenge today is to provide patients with a method that allows them to share the attribute values of their records efficiently and securely while maintaining the privacy preferences
they desire. 
Patients should also be able to selectively share certain parts of their record with the healthcare providers they select. 
In addition, they should not be concerned with the availability nor the guaranteed secure storage of their records. 

\subsection{System Model}
The general model of our proposed record-sharing scheme is illustrated in Fig.~\ref{Fig:model}.
The system consists of six main entities: 

\para{Patients}
Patients are data owners that wish to share their data selectively based on their privacy preferences. 
We denote the set of patients as $\{ p_a \in \calP ~|~ 1 \leq a \leq \infty \}$. 

\para{Medical institutions}
Medical institutions provide medical treatment and generate medical records for the patients.
We denote the set of medical institutions as $\{ m_b \in \mathcal{M} ~|~ 1 \leq b \leq \infty \}$. 

\para{Staff members}
Personnel employed at a medical institution.
We denote the set of staff members at medical institution $m_b$ as $\{ s_{b,l} ~\forall~ 1 \leq l \leq \infty \}$.
Staff members are categorized into groups of similar characteristics based on the set of attributes $\bbA = \{\A_1,\A_2\dots,\A_n\}$ they each possess.
Attributes represent identifiers that are selected from an infinite pool set.
They may be as general as the role of a staff member and could be as unique as a biometric.

\para{Key-issuer}
A semi-trusted party that generates keys for permissioned staff members upon their requests to access parts of the record.
    
\para{Distributed storage P2P network}
A non-trusted P2P network used by the patients to store and share their records with staff members at any medical institution. 

\para{Blockchain P2P network}
A non-trusted P2P network that maintains a blockchain to regulate access permissions of patient records to the staff members of medical institutions.

In general, in order for patients, medical institutions, staff members or key-issuers to interact with any of the two P2P networks, they must run nodes that are capable of connecting to either network.
These nodes may be light-weight nodes (similar to the simple payment verification nodes in Bitcoin~\cite{Nakamoto08Bitcoin}) that can assemble transactions and propagate them to the network.
Running light-weight clients does not require powerful computing machines and can be done via simple machines such as mobile devices making our proposed scheme accessible to everyone.
The common requirement for running either node is being able to generate the public key $\pub$ and private key $\pr$ pair.
\begin{figure}
\centering
\includegraphics[width=\columnwidth]{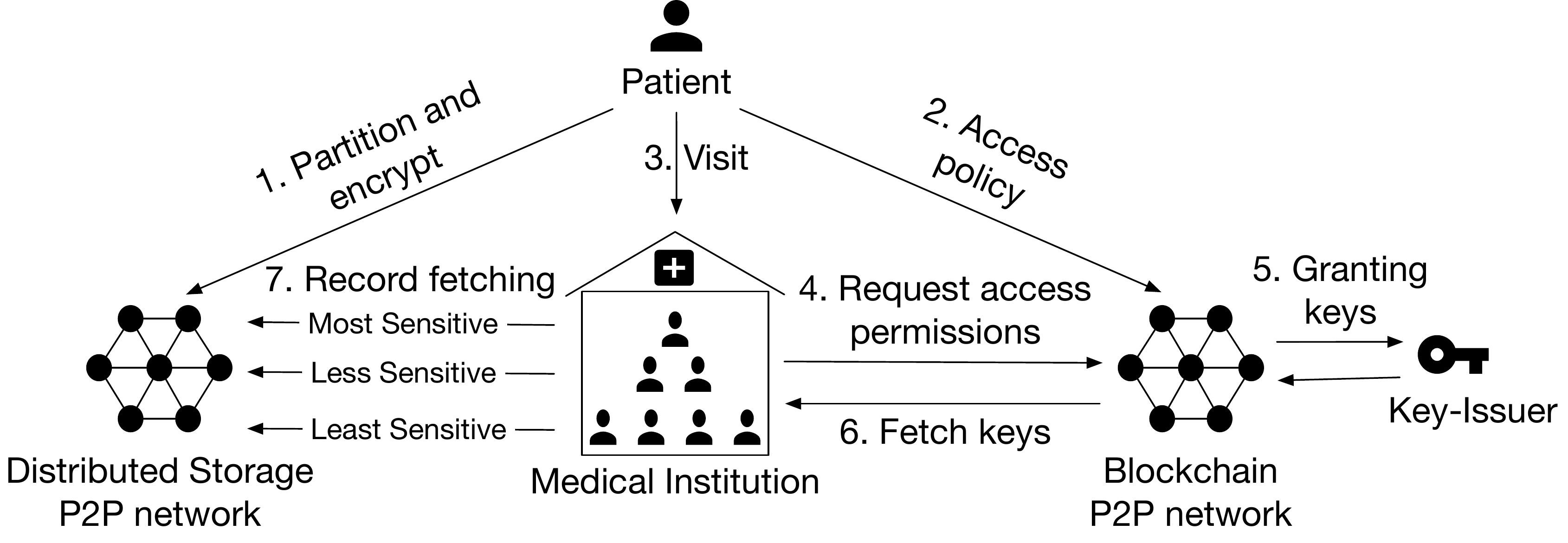}
\caption{General scheme orchestration}
\label{Fig:model}
\end{figure}

\subsection{Design Goals}
Our proposed scheme aims at satisfying the following:

\para{Data Ownership}
Patients are empowered over their records. They have full ownership and control over how to share their records.

\para{Fine-grained access control}
Patients can grant specific access permissions to the desired staff members based on their privacy preferences.
They can selectively share any part of their records using any set of staff member attributes they wish, limiting access to specific parts of their record and to certain staff members at particular medical institutions. 

\para{Collusion resistant}
Two or more staff members that possess different sets of attributes and/or are employed at the same/different institution(s) cannot combine their attributes to gain access to any part of the records they are not authorized to access individually. 

\para{Data confidentiality}
Records are completely protected from any unauthorized staff members that do not possess the correct set of attributes predefined by the patient. 
This includes any type of space that stores the records for the patients.

\para{Distributed storage}
Patient records are stored in a distributed storage network. 
They can be accessed at any time and by any permissioned staff member.
Storage is not centralized and/or controlled by the record-generating institution.

\para{Data access trace-ability}
Patients can trace-back to the entire history of accesses of their records. 
This allows patients to keep track of how their records are being accessed.

\section{The Proposed $d$-MABE Scheme}\label{sec:proposed}
In this section, we first present an overview of the major chain of events in our proposed scheme.
For discussion purposes, we assume that a patient $p_a$ has already received some medical treatment at institution $m_b$ and generated for him/her a medical record $\calR$.
We also assume that the $p_a$ anticipates visiting some other institution $\{ m_b ~|~ b \geq 1 \}$ in the future and would like to share $\calR$ selectively among its staff members.
Based on our description, we then propose a generalized structure through three smart contracts that can be deployed on a blockchain.

\subsection{Scheme Orchestration}
The proposed scheme can be divided into seven major events as demonstrated in Fig.~\ref{Fig:model}.

\para{Record partition and encryption}
Following the work presented in~\cite{zaghloul2018p}, the $p_a$ initially defines a privilege-based access structure that satisfies his/her privacy preferences.
The access structure is divided into $k$ levels $\{\calL_1,\calL_2,\dots,\calL_k\}$ where each level is associated with an access policies $\{\calT_1,\calT_2,\dots,\calT_k\}$.
The $p_a$ then partitions $\calR$ into $k$ segments $\{\R_1,\R_2,\dots,\R_k\}$ of different sensitivity.
Each $\R_i\in \calR$ may represent one or more record attribute(s) depending on how the $p_a$ partitions $\calR$.
Next, the $p_a$ derives a set of symmetric keys $\{\sk_1,\sk_2,\dots,\sk_k\}$ as described in Section~\ref{subsec:pmod}.
Using these keys, the $p_a$ then encrypts each corresponding $\R_i\in \calR$ as
\begin{equation}\label{eq:symenc}
    \ER_i = \symenc_{\sk_i}(\R_i),  
\end{equation}
where $\symenc$ is a symmetric encryption algorithm such as the Advanced Encryption Algorithm (AES)~\cite{daemen2013design}.
Finally, using the $\mathsf{Encryption}$ function discussed in Section~\ref{subsec:cp-abe}, the $p_a$ encrypts each $\sk_i$ under its corresponding $\calT_i$ defined in the privilege-based access structure, such that
\begin{equation}\label{eq:cpabeenc}
\Esk_i = \cpabeenc(\PK,\sk_i,\calT_i),
\end{equation}
where $\cpabeenc$ is the CP-ABE encryption function and $\PK$ is the public key generated during the $\mathsf{Setup}$ phase.
The $p_a$ then stores the generated ciphertexts $\{ \ER_1,\ER_2,\ldots,\ER_k \}$ and $\{ \Esk_1,\Esk_2,\ldots,\Esk_k \}$ over IPFS.
 
\para{Access policy}
To facilitate data management and sharing, the $p_a$ shares the privilege-based access structure that incorporates the access policies $\{\calT_1,\calT_2,\dots,\calT_k\}$.
The access structure is incorporated into a smart contract that is then deployed over the blockchain.

\para{Medical institution visit}
When the $p_a$ visits a medical institution $\{ m_b ~|~ b \geq 1 \}$ to get medical treatment, the $p_a$ interacts with various staff members that may or may not belong to a predefined privilege-based access structure.

\para{Requesting access permissions}
Permissioned staff members will be granted access to certain parts of a record based on the attributes they possess.
To request access, a staff member $s_{b,l}$ interacts with the smart contract previously deployed by the $p_a$ that incorporates the privilege-based access structure.
The smart contract will verify whether the $s_{b,l}$ possesses a set of attributes that can satisfy an access policy within the access structure.
Once the verification is performed the smart contract will publicly announce the access permissions of the $s_{b,l}$ over the blockchain.

\para{Granting keys}
The key-issuer continuously monitors the blockchain for access announcements.
An announcement contains the set of attributes $\bbA$ that the $s_{b,l}$ possesses.
It is a form of verification to the key-issuer that the $s_{b,l}$ possesses set $\bbA$ before generating a secret key $\SK$ using the $\mathsf{Key Generation}$ function described in Section~\ref{subsec:cp-abe}.
The key-issuer then encrypts $\SK$ with the public key $\pub$ of $s_{b,l}$ as the following 
\begin{equation}\label{eq:asymenc}
\ESK = \asymenc_{\pub}(\SK),
\end{equation}
where $\asymenc$ is the asymmetric encryption function.
To share the key with the $s_{b,l}$, the key-issuer transacts with a global smart contract that publicly announces the encrypted key over the blockchain.

\para{Key fetching}
The $s_{b,l}$ can obtain the uniquely encrypted secret key $\ESK$ by monitoring the announcements made over the blockchain.
Once obtained, the $s_{b,l}$ can decrypt $\ESK$ using his/her private key $\pr$ that corresponds to the public key $\pub$ such that
\begin{equation}\label{eq:asymdec}
\SK = \asymdec_{\pr}(\ESK),
\end{equation}
where $\asymdec$ is the asymmetric decryption function.

\para{Record fetching}
Here, the $s_{b,l}$ possesses the necessary key to decrypt the parts he/she has been granted access.
The $p_a$ shares the locations of storage to the encrypted data over IPFS.
Once fetched, the $s_{b,l}$ can decrypt the encrypted symmetric key $\Esk_i$ using the derived key $\SK$ as explained in Section~\ref{subsec:cp-abe}. 
That is
\begin{equation}\label{eq:cpabedec}
    \sk_i = \cpabedec(\Esk_i,\SK),
\end{equation}
where  $\cpabedec$ is the CP-ABE decryption function.
After obtaining $\sk_i$, the $s_{b,l}$ can derive the remaining set of secret keys $\{\sk_{i+1},\dots,\sk_k\}$ using equation~\eqref{keys}.
Finally, the $s_{b,l}$ can decrypt each encrypted partition, such that
\begin{equation}\label{eq:symdec}
    \R_i = \symdec_{\sk_i}(\ER_i),
\end{equation}
where $\symdec$ is the decryption function.

\subsection{Smart contracts}\label{subsec:smartcontracts}
Our proposed scheme includes three main smart contracts: (i)~staff member registration, (ii)~access verification and permission announcements, and (iii)~granting keys.

\begin{algorithm}[t]
\DontPrintSemicolon
\SetAlgoLined

\struct{$\staffMember$}{
  $\staffID$\;
  $\staffAttributes[\upBound]$\;
}

\BlankLine
\variable{$\mathsf{mapping (address} \rightarrow
\staffMember) ~\Map$}

\BlankLine
\tcp{Adding a new staff member}
\Fn{$\addStaffMember (\address, \id, \attributes[\upBound])$}{
    \uIf{$(\msgsender ~\not\in~ \setOfCertifiers )$}{
        \throw \;
    }
    \Else{
        $\Map(\address \rightarrow \id,\attributes[\upBound])$ \;
    }
}

\BlankLine
\tcp{Return the attributes of a staff member}
\Fn{$\getAttributes (\address)$}{
    \return $\Map[\address].\staffAttributes$
}

\caption{$\SMR$ smart contract}
\label{sc1}
\end{algorithm}

\para{Staff member registration}
Staff Member Registration ($\SMR$) is a global smart contract to register staff members of medical institutions that wish to access the patient records.
The process of registering staff members by interacting with this smart contract can be delegated to a number of certified institutions that physically verify staff members against their possessed attributes before registering them. 
This is achieved by incorporating precise policies into the smart contract that require certain identities to trigger the contract.
A certified institution verifies the attributes of a staff member then sends them as input when transacting with the smart contract to be stored over the blockchain.
This process maps the identity (in our case the Ethereum address) of the staff member to the set of attributes possessed.
By observing the blockchain, we can trace-back the on-going activity of these certified institutions as they add, delete, or modify the information of staff members, hence, we can also detect malicious data manipulation.

Algorithm~\ref{sc1} outlines the general functions of the $\SMR$ smart contract.
Lines~1-4 represent a generalized data structure $\staffMember$ that the smart contract uses to store new staff members.
It incorporates the identity $\staffID$ and the array of attributes $\staffAttributes$ of the staff member.
Lines~6-15 represent the two main functions $\addStaffMember$ and $\getAttributes$ of the smart contract.
The $\addStaffMember$ function allows only certified institutions $\setOfCertifiers$ to upload the data of new staff members after physically verifying their attributes.
This conditioned access is to prevent malicious attackers from granting access to themselves or others.
On the contrary, the $\getAttributes$ function can be called by any user. 
Given the address $\address$ of a specific staff member, this function returns the previously verified and stored attributes of this staff member.

\begin{algorithm}[t]
\DontPrintSemicolon
\SetAlgoLined

\Initvar{$\address = \SMR~ \mathsf{address}$}

\event{$\LogAnnounce(\mathsf{address}, \mathsf{attributes}, \calT_i)$}

\BlankLine
\Fn{$\verifyRequest (\msg, \sigma)$}{

    \tcp{Verify 1: validate digital signature}
    $\signer \leftarrow \ecrecover(\msg,\sigma)$\;
    \uIf{($\signer \neq \hash(\pub))$}{
        \throw \;
    }
    \Else{
        \tcp{Verify 2: fetch stored attributes of the staff member}
        $\mathsf{r} \leftarrow \SMR(\address)$\;
        $\fetched \leftarrow \mathsf{r.}\getAttributes(\signer)$\;  
        \tcp{Verify 2: fetched attr. vs access policies}    
        \For{$i\gets1$ \KwTo $k$}{
            \If{($\satisfy(\calT_i,\fetched) == true$)}{
                $\LogAnnounce(\signer, \fetched, \calT_i)$\;
                \textbf{break}
            }    
        }      
    }
}
\caption{$\AVPA$ smart contract}\label{sc2}
\end{algorithm}

\para{Access verification and permission announcements}
The Access Verification and Permission Announcements ($\AVPA$) is a smart contract that incorporates the privilege-based access structure defined by the patients in order to facilitate the selective record management and sharing.
The staff members interested in obtaining parts of the record interact with $\AVPA$ contracts to request access permissions.
The smart contract is outlined in Algorithm~\ref{sc2}.

The smart contract performs a sequential verification process.
This is represented in a single function $\verifyRequest$ that takes two inputs: a message $\msg$ prior to being signed and its signature $\sigma$. That is
\begin{eqnarray}
    \msg = \hash(\staffID),\\
    \sigma = \sign(\pr,\pub,\msg),
\end{eqnarray}
where $\hash$ is the hashing function and $\sign$ is an $\mathsf{ECDSA}$ signing function.

In the initial verification process, the $\AVPA$ smart contract uses an $\mathsf{ECDSA}$ compatible validation function $\ecrecover(\msg,\sigma)$ to validate $\sigma$ by verifying whether
\begin{equation}\label{eq:recover}
    \ecrecover(\msg,\sigma) = \hash(\pub)
\end{equation}
holds true. 
In fact, if the validation function is conducted truthfully, then the correctness of equation \eqref{eq:recover} follows from the $\mathsf{ECDSA}$. 
Next, $\signer$ is compared to the address of the requesting staff member as shown in line~5.
If the addresses match, the contract uses $\signer$ to fetch the previously verified and stored attributes of this staff member in the $\SMR$ contract.
However, this method is liable to replay attacks. 
In section~\ref{sec:security}, we discuss this issue and propose a countermeasure.

If the initial verification is successful, the smart contract executes the second verification as outlined in lines~10-14.
In this process, the $\fetched$ attributes are checked against the access policy $\calT_i$ within the access structure.
The access structure is defined by the patient during contract deployment and maintains the desired privacy settings as discussed in equations~\eqref{eq:symenc} and~\eqref{eq:cpabeenc}.
The access policies are tested sequentially starting from the highest ranked access policy $\calT_1$ at level $\calL_1$.
If the attributes $\satisfy$ the access policy $\calT_i$, the verification is discontinued and the function fires an event $\LogAnnounce$ that announces a permanent log of the smart contract transaction stored over the blockchain.
This event is a public and immutable announcement to ensure that the staff member in possession of $\signer$ should be granted access to the parts of the record $\{ \R_i \ldots \
R_k \}$ corresponding to levels $\{ \calL_i \ldots \calL_k \}$.

\begin{algorithm}[t]
\DontPrintSemicolon
\SetAlgoLined

\event{$\LogKeys(\mathsf{address}, \mathsf{ipfsAddress})$}

\BlankLine
\tcp{Sharing the location of a generated access key}
\Fn{$\addkey (\staffAddress, \ipfsAddress$)}{
    \uIf{($\msgsender ~\not\in~ \setOfIssuers)$}{
        \throw \;    
    }
    \Else{
        $\LogKeys(\staffAddress, \ipfsAddress)$\;
    }
}
\caption{$\GK$ smart contract}
\label{sc3}
\end{algorithm}
\para{Granting keys}
Granting Keys ($\GK$) is a global smart contract that is used to share the location of the generated and encrypted access keys over the distributed storage.
The contract generally consists of a single function, $\addkey$, as outlined in Algorithm~\ref{sc3}.

When the key-issuer observes a $\LogAnnounce$ event fired by the $\AVPA$ smart contract, it generates an access secret key $\SK$ for the specified staff member using the unique set of attributes $\fetched$ stored in the logs.
Next, the key-issuer encrypts this access key with the public key of the staff member as discussed in equation~\eqref{eq:asymenc} and uploads it to the IPFS storage, maintaining its reference location.
Following that, the smart contract fires an event, $\LogKeys$ as shown in line~6, which permanently stores the address of the requesting staff member $\staffAddress$ along with the IPFS storage location $\ipfsAddress$.
Using $\ipfsAddress$, the staff member can fetch the encrypted key $\ESK$ stored over IPFS.
However, as shown in the $\addkey$ function, only certified key-issuers $\setOfIssuers$ are capable of triggering this function.
Therefore, attackers are prevented from spamming the smart contract logs with fake keys.
The staff member can listen for these fired events and obtain the corresponding $\ipfsAddress$ as it appears in the logs.
Using the $\ipfsAddress$, the staff member can fetch $\ESK$ from the IPFS network.
It is important to note that only the staff member in possession of the private key $\pr$ corresponding to the public key $\pub$ will be able to decrypt the fetched encrypted key and obtain $\SK$ as described in equation~\eqref{eq:asymdec}.
Attackers continuously listening to the fired events may be able to learn certain IPFS addresses storing generated encrypted keys, but not the actual keys.
Using the obtained secret key $\SK$, the staff member can then decrypt $\Esk_i$ to generate the secret key $\sk_i$.
Finally, the staff member can decrypt $\ER_i$ stored over IPFS at $\ipfsAddress$ to obtain the record part $\R_i$.

\subsection{Access Permission Revocation}
Attributes possessed by staff members may change over time due to, for example, job switch or retirement.
As a result, the access rights to the patient records should be revoked to be consistent with the access policy.
However, this could be challenging since it requires the attributes of the staff members to be updated periodically.
While adding an expiration date~\cite{pirretti2010secure} to each attribute when registering staff members seems to be a simple solution, it requires the patients to incorporate time constraints into their access policies.

The proposed $d$-MABE scheme can handle this issue efficiently without requiring any extra process for both the patient and the key-issuer.
The staff members only need to renew their registration periodically to ensure freshness of the attribute-based access control.
As a result, if at any point in time a staff member is unable to provide evidence of registration to the level of access claimed, the future access to the patient records will be disabled. 
Consequently, if the staff member attempts to request records he/she is no longer entitled to access, the AVPA smart contract will fail to verify the request.

\section{Security and Privacy Analysis} \label{sec:security}
In this section, we present the security and privacy discussions of our proposed scheme.
We assume our system runs over a blockchain that is maintained by a significant number of nodes, for example, Ethereum. 
In such blockchain platforms, it is very expensive to tamper with verified transactions processed into blocks.
The best bet of the attackers would be to control more than half of the computational power in the network in order to perform 51\% attacks.
We also consider distributed storage networks such as IPFS that are content addressable. 
These networks use the hash computation of the original data when storing and referencing it over the network nodes, resulting in tamper-resistant storage.
Moreover, we consider adversaries that can act as any entity within the system and can manage to connect to either the IPFS or the blockchain network.
Such adversaries are capable of generating fraudulent transactions and propagating them through the blockchain network.
By fraudulent transactions, we mean transacting with smart contracts in an effort to access data they are not granted access.
They may also deploy their own smart contracts over the blockchain, request/store data over the IPFS nodes, and continuously monitor the network channels.

\subsection{Security Analysis}
We first discuss how our proposed scheme is secure against potential attacks.
Next, we present good practices that may help improve the security of the system.

\begin{theorem}
The proposed scheme is secure against replay attacks.
\end{theorem}
\begin{proof}
The initial verification process performed by the $\AVPA$ smart contract requires the requesting staff members to submit $\mathsf{ECDSA}$ digital signatures $\sigma$ to prove their identities.
Given that all data sent over the blockchain is public, malicious attackers can easily maintain a copy of the submitted signatures and later on impersonate the honest staff members giving them access to data they should not be able to access.
To workaround this issue, staff members are required to time-stamp their digital signatures before transacting with the $\AVPA$ contract, such that
\begin{equation}
\sigma \leftarrow \sign(\pr,\pub,\hash(\staffID\Vert \mathsf{T})),
\end{equation}
where $\mathsf{T}$ is the time at which the signature is generated.
Any submitted time-stamped signature is stored over the blockchain in order for the nodes to check if it has been submitted before.
The blockchain nodes will not execute any requests associated with time-stamped signatures that have appeared previously. 
Therefore, to be able to perform replay attacks requires the attackers to reverse the transactions that contain an already used digital signature.
The success of reversing a transaction can be modeled as a race between the honest miners and the attackers trying to generate blocks by competing to solve a hard cryptopuzzle known as Proof-of-Work~\cite{Nakamoto08Bitcoin}.
The race can be modeled as a binomial random walk.
It is denoted as $z$ which represents the number of blocks generated by the honest miners with computational power $p$ minus the number of blocks generated by the attackers with computational power $q = 1 - p$.
This race can be derived as
\begin{equation} \label{eq:random_walk}
z_{i+1} = 
\begin{cases}
z_i + 1, & \text{with probability } p,\\
z_i - 1, & \text{with probability } q.
\end{cases}
\end{equation} 
Using the negative binomial distribution, we can model the probability of success of the attackers.
Assume the key generator waits for $n$ blocks to be generated by the honest miners with computational power $p$ before generating a private key for the requesting attacker. 
Also assume that, at that time, the attacker is able to secretly generate $m$ blocks with computational power $q = 1 - p$, where $m = n - z - 1$.
By definition, this can be modeled as the $m$ number of blocks that the attacker can generate before the $n$ number of blocks generated by the honest miners.
Therefore, the probability of a reversing a transaction for a given value $m$ is 
\begin{equation}
P(m) = \binom{m+n-1}{m}\times p^nq^m.
\end{equation}

Overall, the probability for an attacker to surpass successfully the number of blocks generated by the honest miners can be computed as
\iffinal
\begin{align} \label{eq:sucess2}
\begin{split}
P_s &= \sum_{m=0}^{\infty}P(m) \times Q_{n-m-1} \\
&= 1 - \sum_{m=0}^{\infty}\binom{m+n-1}{m}\times p^nq^m \\
& ~~\times \begin{cases} 
1 - \left(\frac{q}{p}\right)^{n-m}, & \text{if } q < p \text{ or } k \leq n-m, \\
1 - 1 = 0, & \text{if } q > p \text{ or } k > n-m.
\end{cases}
\end{split}
\end{align}
\else
\begin{align} \label{eq:sucess2}
P_s &= \sum_{m=0}^{\infty}P(m) \times Q_{n-m-1} \nonumber\\
&= 1 - \sum_{m=0}^{\infty}\binom{m+n-1}{m}\times p^nq^m \times \begin{cases} 
1 - \left(\frac{q}{p}\right)^{n-m}, & \text{if } q < p \text{ or } k \leq n-m, \\
1 - 1 = 0, & \text{if } q > p \text{ or } k > n-m. 
\end{cases}
\end{align}
\fi

Equation~(\ref{eq:sucess2}) confirms that when $q>p$, the attacker will always succeed since $P_s=1$.
When $q<p$, the probability of success can be defined as
\begin{align} \label{eq:success_graphs}
P_s 
&= 1 - \sum_{m=0}^{n-1}\binom{m+n-1}{k}\times p^nq^m \times \left(1 - \left(\frac{q}{p}\right)^{n-m}\right) \nonumber\\
&= 1 - \sum_{m=0}^{n-1}\binom{m+n-1}{m}\times (p^nq^m - p^mq^n).
\end{align}

Therefore, as more blocks are appended to the blockchain and $q<p$, replay attacks are infeasible.
Moreover, by incorporating active registration, we can limit the pool of malicious attackers to only the staff members with the same level of access, making it even more infeasible. 
\end{proof}

\begin{theorem}
The proposed smart contracts are collusion resistant.
\end{theorem}
\begin{proof}
Our proposed scheme requires staff members to become registered through the $\SMR$ smart contract before being able to request access permissions through the $\AVPA$ contracts.
During registration, each staff member address is mapped to a single set of attributes $\bbA = \{\A_1,\A_2\dots,\A_n\}$ after providing proof of possession to the certified institutions. 
The $\AVPA$ smart contract can only accept a single signature as input when triggered.
Therefore, to perform a collusion attack, the attackers are required to regenerate a single digital signature of an already registered staff member that possesses the set of attributes desired.
For $\mathsf{ECDSA}$, a signature is generated by randomly selecting a cryptographically secure random integer $d$, such that
\begin{equation}\label{eq:curve}
    (x_1,y_1) = d \times G,
\end{equation}
where $(x_1,y_1)$ are the calculated elliptic curve points, $G$ is the generator of the elliptic curve with a large prime order $n$, and $d \in_r [1,n-1]$.
Next, the digital signature is calculated as two components $\sigma = (r,s)$. That is
\begin{equation}\label{eq:r}
r = x_1 ~mod~ n,
\end{equation}
\begin{equation}\label{eq:s}
s = d^{-1}(z + r\cdot pr) ~mod~ n,
\end{equation}
where $z$ is the $L_n$ leftmost bits of $\msg$ such that $L_n$ is the bit length of the group order $n$.
If any of the values $r$ or $s$ are equal to zero, $d$ is randomly selected again and equations~\eqref{eq:curve}, \eqref{eq:r}, and~\eqref{eq:s} are recalculated.
Since the signature components $(r,s)$ are both derived based on the random value $d$, it is infeasible for the attackers to regenerate a specific signature that belongs to a certain staff member, hence, the smart contracts are collusion resistant.
\end{proof}

\begin{theorem}\label{theorem:ipfs}
Using a content-addressable storage such as IPFS ensures that data is tamper-resistant.
\end{theorem}
\begin{proof}
In IPFS, records are initially partitioned into $n$ smaller segments before being stored. That is
\begin{equation}
    \calR = \{\R_1 \ldots \R_n\},
\end{equation}
where each $\R_i \in \calR$ is 256KB, by default.
However, the size of a partition may also be customized by the patient as desired.
Next, each $\R_i \in \calR$ is hashed as
\begin{equation}
    h_i = \hash{(\R_i)},
\end{equation}
where $h_i$ is the resulting digest and is used to reference data stored in the network nodes through the DHT.
For the $\hash$ function, it is computationally infeasible to find an $\R_i' \neq \R_i$, such that 
\begin{equation}
    \hash(\R_i') = \hash(\R_i).
\end{equation}
Therefore, it is computationally infeasible for an attacker to tamper with the records of the patients.
\end{proof}

\begin{theorem}
The proposed scheme can counter Distributed Denial of Service (DDoS) attempts.
\end{theorem}
\begin{proof}
DDoS attacks over IPFS aim at congesting the network by sending numerous requests to the nodes storing the requested data.
However, to perform such attacks, the attackers must initially identify all the nodes storing the target data from the DHT.
Next, they must disrupt the service by sending these nodes a large number of requests, such that
\begin{equation}
    \text{Number of requests} \gg \maxtran,
\end{equation}
where $\maxtran$ is the maximum number of transactions that are accepted by a node as defined in its configuration file.

Attackers may also try to isolate and monopolizes all inbound and outbound connections and data flows of the storage nodes (referred to as an Eclipse attack~\cite{singh2006eclipse}) from the entire network.
These attacks become infeasible as the number of nodes storing the data increases since the record is replicated across the network nodes.
In addition, such attacks are limited by the time $t_{scheme}$ for a request to appear over the blockchain until the staff member fetches the data from the IPFS network nodes.
Attackers must be able to perform the entire attack in time $t_{attack}$. That is 
\begin{equation}
    t_{attack} < t_{scheme}.
\end{equation}
An attacker with no prior knowledge of which data will be requested by staff members has no advantage of identifying the nodes storing the data.
Therefore, as the value of $t_{scheme}$ becomes smaller and with an increased record replication, DDoS attacks become infeasible.
\end{proof}

\subsection{Recommended Security Practices}
Security may also be enhanced by adopting good development practices.
In the following, we present some highly recommended techniques that developers should keep in the back of their minds while implementing the system.

\para{Global smart contracts access control}
It is necessary to ensure integrity of the data stored in the $\SMR$ and $\GK$ global contracts since they provide the necessary information required to verify $\AVPA$ contracts and grant keys to the staff members.
Therefore, it is imperative to halt the attackers from manipulating the state of these contracts with fraudulent information.
Attackers with access rights that can modify the $\SMR$ contract may easily register themselves with any set of attributes required to pass the verification of certain $\AVPA$ contracts.
It is also feasible to perform DDoS attacks by deleting or modifying the registration of staff members in the $\SMR$ contract or spamming the logs of the $\GK$ contract with fraudulent information.
This will result in interrupting or slowing down the process of retrieving patient records.
Therefore, when developing such contracts, it is important to limit the ability to modify the state of global contracts to only certified institutions.
This is outlined in lines~7 and~3 of Algorithms~\ref{sc1} and~\ref{sc3}, respectively.
The current smart contract programming languages, for example, Solidity~\cite{solidity}, provide specific functions to verify certain conditions and/or variable values before execution is performed.
Well known examples include the $\mathsf{require()}$ or $\mathsf{assert()}$ functions that verify preliminary conditions and consume a portion or all of the gas associated with a transaction.
In general, two main reasons for charging users gas to trigger contracts deployed over the blockchain are to reward the executing network nodes and make potential attacks expensive.
This means that attackers must pay for each transaction they request for it to be executed by the network nodes.
This approach will not prevent attackers from registering themselves into the $\SMR$ with a set of fraudulent attributes but could help halt those trying to spam the network.
Assuming an attacker has been able to modify the state of a global contract, the attacker will still have to pay for each requested transaction.
To increase the security measures in this scenario, contracts may be designed to require a minimum gas amount in order to execute them.
However, this countermeasure results in a trade-off between security and cost for the honest users.

\para{External smart contract calls}
Our proposed scheme relies on external smart contract calls that fetch the stored attributes of registered staff members to compare them to those sent by the requesting staff members before validating if they satisfy a certain access policy.
External calls can introduce several unexpected risks or errors if the smart contracts mistakenly execute malicious code.
Therefore, it is important that patients cautiously implement their $\AVPA$ contracts to handle errors when externally calling the $\SMR$ global contracts during the initial verification process.
In Solidity~\cite{solidity}, external calls can be performed in two ways: low-level calls and contract calls.
Low-level calls do not throw exceptions, instead they return false if the external call of another smart contract itself encounters an exception.
Patients that use low-level calls must check if the returned value within the $\AVPA$ contract fails and then properly handle it.
On the other hand, contract calls are more direct as they throw exceptions as encountered by the external call.
From a security standpoint, it may be more secure to use contract calls especially when the patient is less experienced in handling complex scenarios.

\para{Non-Public Blockchains}
The security of our proposed scheme could also be enhanced by running the smart contracts over a consortium or private blockchain.
In such a setting, miners are selected nodes within the peer-to-peer network that are known and trusted.
If any of these nodes attempt to act maliciously, they are immediately identified and eliminated.
A good example of such candidate nodes could be the medical institutions themselves that are interested in maintaining the system in return for efficient data access when required.

\subsection{Privacy Analysis}
While our proposed scheme aims at satisfying the privacy desires of patients and their records, it tends to leak knowledge about the staff members and their access patterns.
Each staff member must initially become registered and is linked to a unique address before engaging with the system.
Since this information can be leaked, attackers can simply monitor requests triggered by the staff members by observing the blockchain activity.
However, since no information is revealed about the records of patients through the $\AVPA$ smart contracts, attackers can only track when staff members trigger requests, but not the data they have requested or the patient information.

From the perspective of patients, our proposed scheme is intentionally designed to allow them to trace-back the accesses performed by staff members to their records.
This is possible by tracing the history of $\LogAnnounce$ events fired as outlined in line~12 of Algorithm~\ref{sc2}.
All fired events are permanently stored over the blockchain allowing the patients to continuously monitor how their records are being accessed.
If a patient notices irregular staff member accesses that are undesired, the patient can simply redeploy a new $\AVPA$ smart contract that defines new or more strict access policies.

\begin{theorem}
Under the proposed model, the privacy of patient records location is preserved.
\end{theorem}
\begin{proof}
First, since the user record is encrypted as described in equation~\eqref{eq:asymenc} before it is uploaded to IPFS, the content of the user record is protected. 
Second, for the current IPFS system, any user in possession of the data can reproduce its cryptographic-hash, search its corresponding location that is maintained by the DHT, and locate it within the peer-to-peer network nodes determined by the default IPFS partitioning techniques.
This will also result in recursively locating any data linked to it. 
To ensure data privacy, we will append a random value $r$ to the record before uploading it to IPFS. 
Therefore, the storage location is determined by
\begin{equation}
    h = \hash(\calR \| r).
\end{equation}
The randomness of $r$ makes it infeasible for the attacker to locate the data stored over IPFS nodes. 
\end{proof}

\section{Performance and Application Analysis} \label{sec:evaluation}
In this section, we will evaluate the proposed smart contracts in terms of performance and monetary costs. 

\subsection{Performance Analysis}
The performance measurement of smart contracts can be performed either through (i) calculating the gas costs required to deploy/transact with the contract\add{,} or (ii) measuring the computational complexity of the algorithms.
Method (i) is usually more accurate since it presents an estimate of the actual monetary costs paid by the users in order to deploy/transact with the contract.
However, this method is directly dependent on the actual source code implementation.
That being said, developers that consider optimization techniques when coding their contracts can probably end up reducing the gas costs.
In addition to this, even if we were to assume access to the entire source code of smart contracts, it is still difficult to determine the expected gas costs, since contracts may behave differently based on their state. 
On the other hand, method (ii) can help give us an indication of the computational complexity based on the number of inputs.
This computational complexity directly correlates to the gas costs paid by the users.
In other words, the more complex an algorithm is, the higher gas costs it will require to be deployed/executed.
However, with smart contracts, the concept of reducing computational complexity to make it run faster is not of significant importance since blockchain transactions are generally executed at discrete intervals.
Therefore, if we were to reduce the computational complexity, it would only be to reduce the gas costs required rather than the total latency.

To measure the performance of our proposed smart contracts, we will present discussions on each of our proposed algorithms that combine both methods.
Our discussions assume that smart contracts will be implemented in Solidity~\cite{solidity}, a contract-oriented, high-level language for implementing smart contracts deployable over the Ethereuem blockchain and processed by the EVM.

\para{$\SMR$ contract}
As outlined in Algorithm~\ref{sc1}, the $\addStaffMember$ function takes $\attributes$ as input to register a new staff member. 
Due to the computational limitations of the EVM, returning dynamic content from external function calls is not possible, i.e. returning a dynamically sized array of attributes resulting from the $\getAttributes$ function when called by the $\AVPA$ contract.
This means, the size of $\Map[\address].\staffAttributes$ must be fixed in order to be executed by the EVM.
As a result, the only workaround is to use statically-sized $\upBound$ arrays of attributes as input to the $\addStaffMember$.
In this case, performance is based on the size of $\upBound$, such that
\begin{equation}
    \text{Performance} \propto \upBound.
\end{equation}

\para{$\AVPA$ contract}
As discussed previously in Agorithm~\ref{sc2}, the $\AVPA$ contract consists of two sequential verifications that play a dominant role in the performance of the contract.
In the initial verification, an input digital signature $\sigma$ is validated to prove the identity of the requesting staff member.
With Solidity, the only available cryptographic function that can perform such a process is the $\mathsf{ECDSA}$ $\ecrecover$ function.
The function recovers the Ethereum address associated with the public key from the elliptic curve signature or returns zero on error.
At the time of writing, the $\ecrecover$ function requires 3000 gas units.
In comparison to most of the available possible computations available by the EVM~\cite{wood2014ethereum}, the $\ecrecover$ function is considered to be relatively expensive.
However, this initial verification is fixed with each $\AVPA$ contract transaction.

Assuming the initial verification is successful, the $\AVPA$ externally fetches the attributes of the staff member to test them against the incorporated access structure.
Similar to the $\SMR$ contract, the performance of the $\AVPA$ contract is dependent on the $\upBound$ of the $\fetched$ attributes.
Finally, the $\fetched$ attributes are tested against the access policies starting at the highest level within the hierarchy.
Here, performance is affected by the complexity of the overall access structure defined and the number of levels $k$ within the hierarchy.
Assuming the worst case scenario, a requesting staff member might have his/her attributes tested against all access policies within the access structure and only receive the least sensitive part $\R_k$ of the record or nothing at all.
The computational complexity can be represented as $\calO(k \times \upBound \times |\calT_i|)$, where $|\calT_i|$ is the number of attributes that form the access policy at level $\calL_i$.
To reduce this complexity, we can modify the $\AVPA$ contract such that the staff members can request that their attributes be tested against only specific access policies rather than being tested sequentially against all policies.
This would reduce the computational complexity to $\calO(\upBound \times |\calT_i|)$.
We can also incorporate optimized search functions, for example, binary search, that can help optimize searching for attributes in the $\fetched$ set against those in the access policies.
This means that we can further reduce the computational complexity to $\calO(\upBound \times \log{|\calT_i|})$.
However, it is important to note that in such a scenario, the patients need to make their access structures publicly available in order for the staff members to request certain access policies.
This results in a trade-off between the performance of the $\AVPA$ contract and the privacy of the access policy defined.

\para{$\GK$ contract}
The $\GK$ contract requires the least computational complexity and gas costs in respect to the $\SMR$ and $\AVPA$ smart contracts.
The computations involved are as simple as firing a $\LogKeys$ event when access permissions are granted to staff members.
The gas costs of such operations are minimal in comparison to the other computations in the $\SMR$ and $\AVPA$ smart contracts.

\subsection{Numerical Results}

In this section, we implement, simulate and present the estimated costs required to deploy/transact with the smart contracts of our proposed scheme in multiple scenarios.
Our experiments are performed over the Ethereum testnet blockchain~\cite{ropsten}.
We specifically choose the Ethereum blockchain given the fact that it is by far the most widely used smart contract hosting blockchain.
However, we note that our proposed scheme can also be implemented over any other blockchain that incorporates smart contracts.

\begin{table}[t!]
    \centering
    \caption{Estimated smart contract deployment with $k=5$ and $N=25$}
    \begin{tabular}{|c|c|c|c|c|}
        \hline 
        \multirow{2}{*}{\textbf{Smart Contract}} & 
        \multicolumn{4}{c|}{\textbf{k=5, N=25}} \\
        \cline{2-5}
        & \textbf{Gas Limit} & \textbf{Cost/ETH} & \textbf{Cost/USD}& \textbf{Data/bytes} \\
        \hline
        \textbf{$\SMR$} & 240383 & 0.004809 & 1.4 & 736\\
        \hline
        \textbf{$\AVPA$} & 857419 & 0.017148 & 5.21 & 3323\\
        \hline
        \textbf{$\GK$} & 125617 & 0.002512 & 0.74 & 304\\
        \hline
    \end{tabular}
    \label{tab:deployment1}
\end{table}

\begin{table}[t!]
    \centering
    \caption{Estimated smart contract deployment with $k=5$ and $N=50$}
    \begin{tabular}{|c|c|c|c|c|}
        \hline 
        \multirow{2}{*}{\textbf{Smart Contract}} & 
        \multicolumn{4}{c|}{\textbf{k=5, N=50}} \\
        \cline{2-5}
        & \textbf{Gas Limit} & \textbf{Cost/ETH} & \textbf{Cost/USD}& \textbf{Data/bytes} \\
        \hline
        \textbf{$\SMR$} & 240447 & 0.004809 & 1.42 & 736\\
        \hline
        \textbf{$\AVPA$} & 1132628 & 0.022653 & 6.67 & 4536\\
        \hline
        \textbf{$\GK$} & 125617 & 0.002512 & 0.74 & 304\\
        \hline
    \end{tabular}
    \label{tab:deployment2}
\end{table}

\begin{table}[t!]
    \centering
    \caption{Estimated smart contract deployment with $k=10$ and $N=100$}
    \begin{tabular}{|c|c|c|c|c|}
        \hline 
        \multirow{2}{*}{\textbf{Smart Contract}} & 
        \multicolumn{4}{c|}{\textbf{k=10, N=100}} \\
        \cline{2-5}
        & \textbf{Gas Limit} & \textbf{Cost/ETH} & \textbf{Cost/USD}& \textbf{Data/bytes} \\
        \hline
        \textbf{$\SMR$} & 240447 & 0.004809 & 1.42 & 736\\
        \hline
        \textbf{$\AVPA$} & 1303607 & 0.026072 & 7.56 & 5188\\
        \hline
        \textbf{$\GK$} & 125617 & 0.002512 & 0.74 & 304\\
        \hline
    \end{tabular}
    \label{tab:deployment3}
\end{table}

In our simulations, we favor privacy over performance, hence, our implemented $\AVPA$ smart contracts sequentially check whether the attributes of the requesting users satisfy the access policies in order starting from the highest level of the hierarchy.
Given these circumstances, we note that our presented results can be further enhanced in terms of cost as discussed previously.
However, we intentionally choose to implement our smart contracts this way to show that even with these conditions, our proposed contracts are still cheaper when compared to the current systems used by the record-generating institutions to share medical records.
In contrast to the current record-sharing systems, our proposed scheme eliminates the role of the record-generating institution completely. 
This results in eliminating other overhead costs required when sharing medical records.
For example, in developed countries such as the U.S., record-generating institutions must comply with certain state laws that enforce a maximum fee a record-generating institution may charge when requested to share its records.
Table~\ref{tab:state_fees} illustrates a sample of these fees that could be inclusive or exclusive to the methods of delivering the record itself.
Given the current competitiveness, in most cases, the medical institutions would charge the entire allowed fees to maximize their profits.

For our simulations, we apply the values $k = \{ 5,10 \}$, $\upBound = \{ 10,20,30,40,50 \}$ and $N = \{ 25,50,100\}$, where $k$ is the number of levels in the hierarchy and $N$ is the total number of attributes used in the entire access structure.
Without loss of generality, we also assume that the number of attributes for each level follows the uniform distribution.
For example, if $k=5$ and $N=50$, then the number of attributes used to define an access policy is $|\calT_i| = 10$.
The following experiments have been conducted through an Ethereum node running on a system with 1.8 GHz Intel Core i5 and 8 GB RAM.
Our numerical results are also the averages of 10 trials under each scenario.

Based on our experiments, the costs of deploying our smart contracts are not affected by the value of $\upBound$ since there is no major change in code as this value changes.
Therefore, for all values of $\upBound$ that we tested, the costs remained constant.
Tables~\ref{tab:deployment1},\ref{tab:deployment2}, and \ref{tab:deployment3} summarize these costs in terms of gas limits and their estimated and equivalent costs in ETH and United States Dollar (USD) at the time of testing.
To measure the optimum gas limit for each smart contract deployment, we used the JSON-RPC method, $\mathsf{eth\_estimateGas}$~\cite{gaslimit}, that generates and returns an estimate of the required gas limit based on the network success rate.
We also set the gas price to 20 GWEI, where 1 ETH = $1 \times 10^9$ GWEI.
We intentionally select this value relatively higher than the average gas prices specified in~\cite{etherscan} for guaranteed and fast processing.
Again, we note that our presented results can be further optimized by selecting reduced gas price values.
As demonstrated in Tables~\ref{tab:deployment1}, \ref{tab:deployment2}, and \ref{tab:deployment3}, the estimated costs for deploying the $\SMR$ and $\GK$ smart contracts remain constant as we alter the values of $k$ and $N$.
However, the costs for deploying the $\AVPA$ smart contracts increase with an increase in the values $k$ and/or $N$.

\setlength{\tabcolsep}{1pt}
\begin{table}[t]
{
\centering
\caption{U.S maximum state fees~\cite{STATE-BY-STATE-Fees} to copy and deliver records upon being requested by others.}
\label{tab:state_fees}
 \begin{tabular}{|m{1in}|m{1in}|m{1.5in}|m{2in}|} \hline
 \textbf{State} & \textbf{Search fee} & \textbf{Cost per page} & \textbf{Misc. costs per page}\\ \hline
 \textbf{California}  & N/A & \$0.25 & Microfilm: \$0.50\\ \hline
\textbf{Texas} & \$82.95 & \shortstack[l]{P.1-10: \$45.79 flat fee\\ P.11-60: \$1.54\\ P.61-400: \$0.76\\ P.401+: \$0.41} & \shortstack[l]{Microfilm:\\ P.1-10: \$69.74 flat fee\\ P.11+: \$1.54}\\ \hline  
 \textbf{Florida} & \$1.00/year & \$1.00 & Microfilm: \$2.00 \\ \hline  
 \textbf{New York} & N/A & \$0.75 & X-rays: Actual cost of reproduction \\ \hline
 \textbf{Illinois} & \$27.91  & \shortstack[l]{P.1-25: \$1.05\\ P.26-50: \$0.70\\ P.50+: \$0.35} & Microfilm: \$1.74\\ \hline
\end{tabular}
}
\end{table} 

\setlength{\tabcolsep}{4.5pt}
\begin{table}[t] 
\centering
\caption{Estimated $\addStaffMember$ function costs} 
\label{tab:addStaff}
\begin{tabular}{|c|c|c|c|c|c|}
\hline 
 & 
\multicolumn{5}{c|}{\textbf{upBound}}\\
\cline{2-6}
& \textbf{10} & \textbf{20} & \textbf{30} & \textbf{40} & \textbf{50}\\
\hline
\textbf{Gas Limit} & 246531 & 449890 & 653249 & 856674 & 1060034 \\
\hline
\textbf{Cost/ETH} & 0.004931 & 0.008998 & 0.013065 & 0.017133 & 0.021201 \\
\hline
\textbf{Cost/USD} & 1.45 & 2.65 & 3.85 & 5.06 & 6.29\\
\hline
\end{tabular}
\end{table}

\setlength{\tabcolsep}{4.5pt}
\begin{table}[t!] 
\centering
\caption{Estimated $\addkey$ function costs} 
\label{tab:addKey}
\begin{tabular}{|c|c|c|c|c|c|}
\hline 
 & 
\multicolumn{5}{c|}{\textbf{upBound}}\\
\cline{2-6}
& \textbf{10} & \textbf{20} & \textbf{30} & \textbf{40} & \textbf{50}\\
\hline
\textbf{Gas Limit} & 26379 & 26380 & 26375 & 26379 & 26378 \\
\hline
\textbf{Cost/ETH} & 0.000528 & 0.000528 & 0.00053 & 0.000524 & 0.000528 \\
\hline
\textbf{Cost/USD} & 0.14 & 0.14 & 0.14 & 0.15 & 0.15\\
\hline
\end{tabular}
\end{table}

\begin{figure}[t]
\centering
\includegraphics[width=.9\columnwidth]{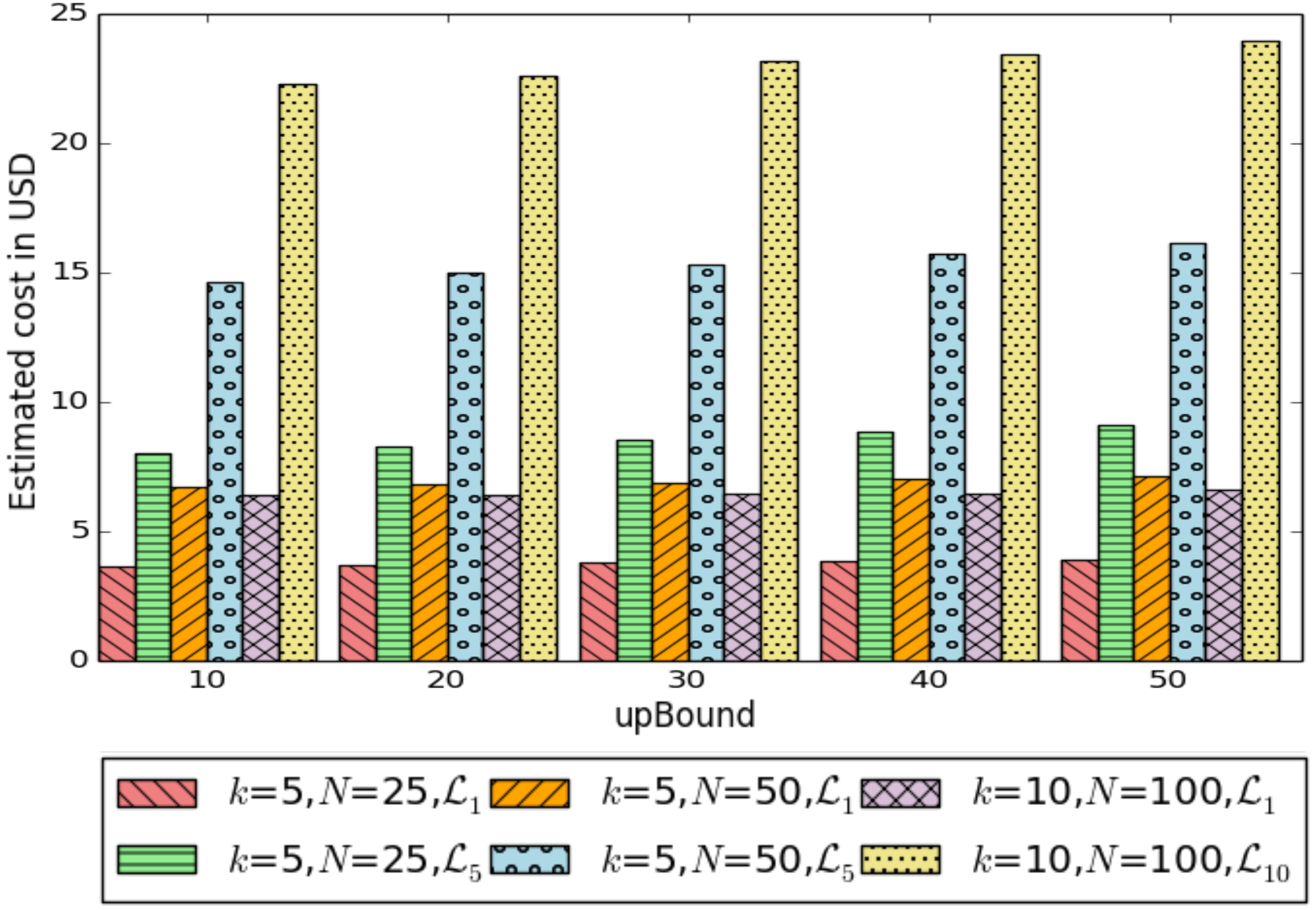}
\caption{Estimated costs to run the $\AVPA$ smart contract}
\label{Fig:avpa_transaction}
\end{figure}
Figure~\ref{Fig:avpa_transaction} demonstrates the changes in USD costs of transacting with already deployed $\AVPA$ contracts as we alter the values $k$ and $N$ for permissioned staff members at levels $\calL_1$, $\calL_5$ and $\calL_{10}$.
As shown by the figure, we recognize an increase in costs as these values increase.
We also realize that as the $\upBound$ value increases while keeping the values $k$, $N$, and $\calL_i$ constant, the costs slightly increase.

Finally, in Tables~\ref{tab:addStaff} and~\ref{tab:addKey}, we present the gas limits and costs in both ETH and USD to transact with the $\addStaffMember$ and $\addkey$ functions.
As shown in Table~\ref{tab:addStaff}, the costs of the transactions increase as the value $\upBound$ increases.
On the contrary, as presented in Table~\ref{tab:addKey}, the costs of transacting with the $\addkey$ function remains constant regardless of the $\upBound$ value used.

\subsection{Extended Application Discussions}
Aggregated patient data has the capability of transforming the future standard of care for patients through the application of predictive analytics and big data methods.
One of the biggest challenges to using health data to its fullest extent is the inaccessibility and segmentation of EMRs~\cite{white2014review}.
As demonstrated by our analyses, the proposed scheme can eliminate these constraints.
It allows healthcare providers the ability to efficiently extract knowledge from disconnected EMR sources that have been willingly shared.

For example, a meaningful opportunity for big data analytics in this context is the capability to model drug and treatment efficacy.
Healthcare providers can access previously shared EMRs to understand how life-saving drugs and treatments perform, respective to the disease state and profile of the patient.
This data can be used to make enhanced and customized treatment decisions for patients.
As healthcare treatment becomes more individualized, successful patient outcomes are more likely and a one-size-fits-all approach to patient treatment is no longer adequate.

\section{Conclusion} \label{sec:conclusion}
In this paper, we proposed $d$-MABE, a secure, privacy-preserving and distributed data sharing scheme that runs over a blockchain using smart contracts.
We demonstrated that $d$-MABE can be used in cases such as medical record-sharing where patients require an efficient and selective method to share their records with staff members of medical institutions.
Our $d$-MABE proposed scheme eliminates the reliance on the record-generating institutions when data is shared and completely empowers the patients.
Our security and privacy analyses show that $d$-MABE is secure and preserves the privacy of patient records.
We also presented some recommended security practices developers should keep in mind when developing their system.
Finally, our comprehensive evaluation proves the efficiency of $d$-MABE and presents numerical results that support our performance analysis.

\bibliographystyle{ieeetr}
\bibliography{smartcontracts}

\begin{thebibliography}{10}

\bibitem{fernandez2013security}
J.~L. Fern{\'a}ndez-Alem{\'a}n, I.~C. Se{\~n}or, P.~{\'A}.~O. Lozoya, and
  A.~Toval, ``Security and privacy in electronic health records: A systematic
  literature review,'' {\em J. of biomedical informatics}, vol.~46, no.~3,
  pp.~541--562, 2013.

\bibitem{gillum2013papyrus}
R.~F. Gillum, ``From papyrus to the electronic tablet: a brief history of the
  clinical medical record with lessons for the digital age,'' {\em The American
  journal of medicine}, vol.~126, no.~10, pp.~853--857, 2013.

\bibitem{hhs}
U.~D. of~Health and H.~Services, ``Heath insurance portability and
  accountability act.'' \url{https://www.hhs.gov/hipaa/}, 2018.

\bibitem{reuters}
L.~Rapaport, ``Few {U.S.} hospitals can fully share electronic medical
  records.''
  https://www.reuters.com/article/us-health-medicalrecords-sharing/few-u-s-hospitals-can-fully-share-electronic-medical-records-idUSKCN1C72UV,
  2017.

\bibitem{lau2000distributed}
F.~Lau, S.~H. Rubin, M.~H. Smith, and L.~Trajkovic, ``Distributed denial of
  service attacks,'' in {\em Systems, Man, and Cybernetics, 2000 IEEE
  International Conference on}, vol.~3, pp.~2275--2280, IEEE, 2000.

\bibitem{anderson2008security}
R.~Anderson, {\em Security engineering}.
\newblock John Wiley \& Sons, 2008.

\bibitem{becker}
J.~Spitzer, ``The cost of a data breach in healthcare averages \$717k: 5 report
  findings.'' https://www.beckershospitalreview.com/, 2018.

\bibitem{Nakamoto08Bitcoin}
S.~Nakamoto, ``Bitcoin: A p2p electronic cash system.''
  \url{https://bitcoin.org/bitcoin.pdf}, 2008.

\bibitem{szabo1996smart}
N.~Szabo, ``Smart contracts: building blocks for digital markets,'' {\em
  EXTROPY: The Journal of Transhumanist Thought,(16)}, 1996.

\bibitem{wood2014ethereum}
G.~Wood, ``Ethereum: A secure decentralised generalised transaction ledger,''
  {\em Ethereum project yellow paper}, vol.~151, pp.~1--32, 2014.

\bibitem{zaghloul18}
E.~Zaghloul, T.~Li, and J.~Ren, ``An attribute-based distributed data sharing
  scheme,'' in {\em 2018 IEEE Global Communications}, IEEE, 2018.

\bibitem{zyskind2015decentralizing}
G.~Zyskind, O.~Nathan, {\em et~al.}, ``Decentralizing privacy: Using blockchain
  to protect personal data,'' in {\em Security and Privacy Workshops (SPW),
  2015 IEEE}, pp.~180--184, IEEE, 2015.

\bibitem{azaria2016medrec}
A.~Azaria, A.~Ekblaw, T.~Vieira, and A.~Lippman, ``Medrec: Using blockchain for
  medical data access and permission management,'' in {\em Open and Big Data
  (OBD)}, pp.~25--30, IEEE, 2016.

\bibitem{dubovitskaya2017secure}
A.~Dubovitskaya, Z.~Xu, S.~Ryu, M.~Schumacher, and F.~Wang, ``Secure and
  trustable electronic medical records sharing using blockchain,'' in {\em AMIA
  Annual Symposium Proceedings}, vol.~2017, p.~650, American Medical
  Informatics Association, 2017.

\bibitem{dinh2018untangling}
T.~T.~A. Dinh, R.~Liu, M.~Zhang, G.~Chen, B.~C. Ooi, and J.~Wang, ``Untangling
  blockchain: A data processing view of blockchain systems,'' {\em IEEE
  Transactions on Knowledge and Data Engineering}, vol.~30, no.~7,
  pp.~1366--1385, 2018.

\bibitem{dinh2017blockbench}
T.~T.~A. Dinh, J.~Wang, G.~Chen, R.~Liu, B.~C. Ooi, and K.-L. Tan,
  ``Blockbench: A framework for analyzing private blockchains,'' in {\em
  Proceedings of the 2017 ACM International Conference on Management of Data},
  pp.~1085--1100, ACM, 2017.

\bibitem{vo2018research}
H.~T. Vo, A.~Kundu, and M.~K. Mohania, ``Research directions in blockchain data
  management and analytics.,'' in {\em Extending Database Technology EDBT},
  pp.~445--448, 2018.

\bibitem{bethencourt2007ciphertext}
J.~Bethencourt, A.~Sahai, and B.~Waters, ``Ciphertext-policy attribute-based
  encryption,'' in {\em 2007 IEEE symposium on security and privacy (SP'07)},
  pp.~321--334, IEEE, 2007.

\bibitem{zaghloul2018p}
E.~Zaghloul, K.~Zhou, and J.~Ren, ``P-mod: Secure privilege-based multilevel
  organizational data-sharing in cloud computing,'' {\em arXiv preprint
  arXiv:1801.02685}, 2018.

\bibitem{buterin2017ethereum}
V.~Buterin, ``Ethereum: A next-generation smart contract and decentralized
  application platform, 2013.'' \url{http://ethereum. org/ethereum. html},
  2017.

\bibitem{baumgart2007s}
I.~Baumgart and S.~Mies, ``S/kademlia: A practicable approach towards secure
  key-based routing,'' in {\em Parallel and Distributed Systems, 2007
  International Conference on}, pp.~1--8, IEEE, 2007.

\bibitem{ipfs}
J.~Benet, ``{IPFS} - content addressed, versioned, {P2P} file system (draft
  3).''
  \url{https://github.com/ipfs/papers/raw/master/ipfs-cap2pfs/ipfs-p2p-file-system.pdf}.

\bibitem{freedman2004democratizing}
M.~J. Freedman, E.~Freudenthal, and D.~Mazieres, ``Democratizing content
  publication with coral.,'' in {\em NSDI}, vol.~4, pp.~18--18, 2004.

\bibitem{bittorrent}
``Bittorrent.'' \url{http://www.bittorrent.com}, 2018.

\bibitem{dag}
T.~Cipriani, ``Visualizing git's merkle dag with d3.js.''

\bibitem{git}
``Git.'' \url{https://git-scm.com/docs}, 2018.

\bibitem{mazieres2000self}
D.~D.~F. Mazi{\`e}res, {\em Self-certifying file system}.
\newblock PhD thesis, Massachusetts Institute of Technology, 2000.

\bibitem{daemen2013design}
J.~Daemen and V.~Rijmen, {\em The design of Rijndael: AES-the advanced
  encryption standard}.
\newblock Springer Science \& Business Media, 2013.

\bibitem{pirretti2010secure}
M.~Pirretti, P.~Traynor, P.~McDaniel, and B.~Waters, ``Secure attribute-based
  systems,'' {\em J. of Computer Security}, vol.~18, no.~5, pp.~799--837, 2010.

\bibitem{singh2006eclipse}
A.~Singh, T.~wan Ngan, P.~Druschel, and D.~S. Wallach, ``Eclipse attacks on
  overlay networks: Threats and defenses,'' in {\em In IEEE INFOCOM}, Citeseer,
  2006.

\bibitem{solidity}
``Solidity.'' \url{http://solidity.readthedocs.io/en/v0.4.24/}, 2018.

\bibitem{ropsten}
``Ropsten (revival) testnet.'' \url{https://ropsten.etherscan.io}, 2018.

\bibitem{gaslimit}
``Json {RPC}.''
  \url{https://github.com/ethereum/wiki/wiki/JSON-RPC#eth_estimategas}, 2018.

\bibitem{etherscan}
``Ethereum average gasprice chart.'' \url{https://etherscan.io/chart/gasprice},
  2018.

\bibitem{STATE-BY-STATE-Fees}
MediCopy, ``State-by-state guide of medical record copying fees.''
  \url{https://medicopy.net/who-we-are/blog/guide-of-state-statutes-for-copies-of-medical-records},
  2018.

\bibitem{white2014review}
S.~E. White, ``A review of big data in health care: challenges and
  opportunities,'' {\em Open Access Bioinformatics}, vol.~6, pp.~13--18, 2014.

\end{thebibliography}

\end{document}